\documentclass{llncs}

\usepackage[latin1]{inputenc}
\usepackage{amsmath,amssymb,amsfonts}

\newcommand{\suppress}[1]{}

\newcommand{\reduceitem}{}
\newcommand{\reduceitemm}{}
\newcommand{\centermath}[1]{\smallskip\\\centerline{#1}\smallskip}

\newsavebox{\fmbox}
\newenvironment{fmpage}[1]
     {\medskip\begin{lrbox}{\fmbox}\begin{minipage}{#1}}
     {\end{minipage}\end{lrbox}\fbox{\usebox{\fmbox}}\medskip}
\newcommand{\encadre}[1]{
\begin{center}
\begin{fmpage}{11cm}
#1
\end{fmpage}
\end{center}}

\newcommand{\e}{{\mathrm{e}}}
\newcommand{\complexi}{{\mathrm{i}}}

\newcommand{\complex}{\mathbb{C}}
\newcommand{\Span}{\mathsf{Span}}
\newcommand{\ket} [1] {\lvert #1 \rangle}
\newcommand{\bra} [1] {\langle #1 \rvert}
\newcommand{\braket}[2]{\langle #1\vert #2 \rangle}

\newcommand{\norm}[1]{\lVert #1 \rVert}
\newcommand{\size}[1]{\left\lvert #1 \right\rvert}
\newcommand{\abs}{\size}

\newcommand{\tensor}{\otimes}
\newcommand{\isomorphic}{\cong}
\newcommand{\diag}{{\mathrm{diag}}}
\newcommand{\Order}{{\mathrm{O}}}
\newcommand{\order}{{\mathrm{o}}}

\newcommand{\ancilla}{\ket{\bar{0}}}

\newcommand{\comment}[1]{}

\newcommand{\reflex}{\mathrm{ref}}

\newcommand{\id}{\mathrm{Id}}
\newcommand{\eps}{\varepsilon}
\newcommand{\aitch}{{\mathcal{H}}}
\newcommand{\ay}{{\mathcal{A}}}
\newcommand{\bee}{{\mathcal{B}}}
\newcommand{\beee}{{\mathcal{B'}}}
\newcommand{\ess}{{\mathcal{S}}}
\newcommand{\emm}{{\mathcal{M}}}
\newcommand{\kay}{{\mathcal{K}}}

\newcommand{\beq}{\begin{equation}}
\newcommand{\eeq}{\end{equation}}
\newcommand{\beqa}{\begin{eqnarray}}
\newcommand{\eeqa}{\end{eqnarray}}

\begin{document}

\title{Quantum walk based search algorithms
\thanks{
Research supported by the European Commission IST Integrated Project
Qubit Applications (QAP) 015848, and by the 
ANR Blanc AlgoQP grant of the French Research Ministry. 
}}

\author{
Miklos Santha\inst{1,2}
}

\institute{
{ 
CNRS--LRI, Universit\'e Paris--Sud,
91405 Orsay, France }
\and
{Centre for Quantum Technologies,
Nat. Univ. of Singapore, Singapore 117543
}
}

\maketitle

\begin{abstract}
In this survey paper we  give an intuitive treatment of the discrete time quantization
of classical Markov chains. Grover search and the quantum walk based search algorithms of 
Ambainis, Szegedy and Magniez et al. will be stated as quantum analogues of classical search
procedures. We present a rather detailed description of a somewhat simplified version of the MNRS 
algorithm. Finally,  in the query complexity model, we show how quantum walks can be applied to
the following search problems: Element Distinctness,
Matrix Product Verification, Restricted Range Associativity, Triangle, and Group Commutativity.

\end{abstract}

\section{Introduction}
Searching is without any doubt one of the major problems in computer science.
The corresponding literature is tremendous, most manuals on algorithms
include several chapters that deal with searching procedures~\cite{Knu,CLRS}.
The relevance of finite Markov chains (random walks in graphs) to searching was recognized
from early on, and it is still a flourishing field.
The algorithm of Aleliunas et al.~\cite{AKLLR} 
that solves $s$--$t$ connectivity in undirected
graphs in time $O(n^3)$ and in space $O(\log n)$, and Sch\"oning's algorithm~\cite{Sch} 
that provides the basis of the currently fastest solutions for 3-SAT
are among the most prominent examples for that.

Searching is also a central piece in the emerging field of quantum algorithms.
Grover search~\cite{Grover96}, and in general amplitude amplification~\cite{BrassardHMT02}
are well known quantum procedures which are provably faster than their classical counterpart. 
Grover's algorithm was used recursively
by Aaronson and Ambainis~\cite{AaronsonA05} for searching in grids.

Discrete time quantum walks were introduced gradually by  
Meyer\cite{Meyer96,Meyer96b} in connection 
with cellular automata, and by Watrous in his works related to 
space bounded computations~\cite{Watrous01}. Different parameters related to quantum walks
and possible speedups of classical algorithms were investigated by several 
researchers~\cite{NayakV00,AmbainisBNVW01,AharonovAKV01,MR,Kem1,ric07}.

The potential of discrete time quantum walks with respect to searching problems was first pointed out
by Shenvi, Kempe, and Whaley~\cite{ShenviKW03} 
who designed a quantum walk based
simulation of Grover search. Ambainis, in his seminal paper~\cite{Ambainis04}, used
quantum walks on the Johnson graphs to settle the query complexity
of the Element Distinctness problems. Inspired by the work of Ambainis, 
Szegedy~\cite{Szegedy04} designed a general method to quantize
classical Markov chains, and developed a theory of quantum walk
based search algorithms. 
A similar approach for the specific case of searching in grids was taken
by Ambainis, Kempe and Rivosh~\cite{AmbainisKR05}.
The frameworks of Ambainis and Szegedy were used in various contexts
to find
algorithms  with substantial complexity gains over simple Grover
search~\cite{MagniezSS05,BuhrmanS06,MagniezN05,DT}.
In a recent work, Magniez, Nayak, Roland and Santha~\cite{MNRS}
proposed a new quantum walk based search method that expanded the scope of the previous 
approaches. The MNRS search algorithm is also conceptually simple, and improves
various aspects of many walk based algorithms.

In this survey paper we  give an intuitive (though formal) treatment of the quantization
of classical Markov chains. We will be concerned with discrete time quantum walks, the
continuous case will not be covered here.
Grover search and the quantum walk based search algorithms of 
Ambainis, Szegedy and Magniez et al. will be stated as quantum analogues of classical search
procedures. We present a rather detailed description of a somewhat simplified version of the MNRS 
algorithm. Finally,  in the query complexity model, we show how quantum walks can be applied to
the following search problems: Element Distinctness,
Matrix Product Verification, Restricted Range Associativity, Triangle, and Group Commutativity.
For a detailed introduction to quantum walks the reader is referred to the excellent surveys
of Kempe~\cite{Kem2} and Ambainis~\cite{Amb1}. Another survey on quantum search algorithms is 
also due to Ambainis~\cite{Amb2}.

\section{Classical search algorithms}
\label{classical}
At an abstract level, any search problem may be cast as the problem
of finding a marked element from a set~$X$. 
Let~$M \subseteq X$ be the set of marked elements, and let
$\eps$ be a known lower bound on $|M|/|X|$, the fraction of marked elements,
whenever $M$ is non-empty. 
If no further information is available on $M$, we can choose $\eps$ 
as $1/|X|$. The simplest approach, stated in
{\bf Search Algorithm~1}, to solve this problem
is to repeatedly
sample from~$X$ uniformly until a marked element is picked, if there is any.

\encadre{ {\bf Search Algorithm~1} \\
Repeat for~$t = O(1/ \eps )$ steps
\begin{enumerate}\reduceitem
\item Sample~$x \in X$ according to the uniform distribution.
\item 
If $x$ is in $M$ then output it and stop.
\end{enumerate}
} 
More sophisticated approaches might use a Markov chain on the state space $X$
for generating the samples.
In that case, to generate the next sample, the resources expended for 
previous generations are often reused.

Markov chains can be viewed as random walks on directed graphs with weighted edges.
We will identify a Markov chain with its transition matrix $P = (p_{xy}).$ A chain is 
{\em irreducible} if every state is accessible from every other state.
An irreducible chain is {\em ergodic} if it is also aperiodic.
The eigenvalues of a Markov chain are at most 1 in magnitude.
By the Perron-Frobenius
theorem, an irreducible chain has a unique stationary distribution~$\pi=(\pi_x)$,
that is a unique left eigenvector~$\pi$ with eigenvalue~$1$ and
positive coordinates summing up to~$1$. If the chain is ergodic, 
the eigenvalue~$1$ is the only
eigenvalue of $P$ with magnitude $1$.
We will denote by $\delta = \delta(P)$ the {\em eigenvalue gap} of $P$, that is $1- | \lambda|$,
where $\lambda$ is an eigenvalue with the second largest magnitude.
It follows that when $P$ is ergodic then $\delta > 0$.

The {\em time-reversed Markov chain} $P^*$ 
of $P$ is defined by the
equations $\pi_x p_{xy} = \pi_y p^*_{yx}$. 
The Markov chain $P$ is said to be {\em
reversible\/} if $P^* = P$. 
Reversible chains can be viewed as random walks on undirected graphs with weighted edges,
and in these chains the
probability of a transition from a state~$x$ to another state~$y$ in
the stationary distribution is the same as the probability of the
transition in the reverse direction.
The Markov chain $P$ is {\em symmetric} if $P = P^t$ where $P^t$
denotes the transposed matrix of $P$. 
The stationary distribution of symmetric chains is the uniform distribution.
They can be viewed as random walks on regular graphs, and
they are time-reversible.

We consider two search algorithms based on some ergodic and symmetric chain $P$.
{\bf Search Algorithm~2} repeatedly samples from approximately stationary distributions,
and checks if the element is marked.
To get a sample the Markov chain is simulated long enough to mix well.  
{\bf Search Algorithm~3} is a 
greedy variant: a check is performed after every step of the chain.

\encadre{ {\bf Search Algorithm~2}
\mbox{ } 
\begin{enumerate}\reduceitemm

\item Initialize~$x$ to a state sampled from the uniform
distribution over~$X$.

\item Repeat for~$t_2 = O(1/ \eps)$ steps
\begin{enumerate}\reduceitemm

\item 
If the element reached in the previous step is marked then output it and stop.

\item Simulate~$t_1 = O(1/ \delta)$ steps of $P$ starting with $x$.

\end{enumerate}

\end{enumerate}
}

\encadre{ {\bf Search Algorithm~3}
\mbox{ } 
\begin{enumerate}\reduceitem

\item Initialize~$x$ to a state sampled from the uniform
distribution over~$X$.

\item Repeat for~$t = O( 1/ \eps \delta )$ steps
\begin{enumerate}\reduceitemm

\item 
   If the element reached in the previous step is marked then output it and stop.

\item Simulate one step of $P$ starting with $x$.

\end{enumerate}

\end{enumerate}
} 

We state formally the complexity of the three algorithms to clarify
their differences.  
They will
maintain a data structure~$d$ that associates some
data~$d(x)$ with every state $x \in X$.  
Creating and maintaining the data structure incurs a certain cost, but
the data $d(x)$ can be helpful
to determine if $x \in M$.  
We distinguish three types of cost.
\begin{itemize}\reduceitem

\item[] \textbf{Setup cost $\mathsf{S}$:} The  cost to
  sample~$x \in X$ according to the uniform distribution, and to
  construct~$d(x)$.

\item[] \textbf{Update cost $\mathsf{U}$:} The cost to
  simulate a transition from~$x$ to~$y$ 
  according to~$P$, and to update $d(x)$ to~$d(y)$.

\item[] \textbf{Checking cost $\mathsf{C}$:} The
cost of checking if~$x \in M$ using~$d(x)$.

\end{itemize}
The cost may be thought of as a vector listing all the measures of
complexity of interest, such as query and time complexity.
The generic bounds on the efficiency of the three search
algorithms can be stated in terms of the cost parameters.

\begin{proposition}
\label{thm-classical-search}
Let $P$ be an ergodic and symmetric Markov chain on $X$.
Then all three algorithms find a marked element with high probability if there is any.
The respective costs incurred by the algorithms are of the following order:
\begin{enumerate}
\item {\bf Search Algorithm~1}:
$
(\mathsf{S}+\mathsf{C})/ \eps,
$

\item {\bf Search Algorithm~2}:
$
\mathsf{S} + 
(\mathsf{U} / \delta +\mathsf{C} ) / \eps,
$

\item {\bf Search Algorithm~3}:
$
\mathsf{S} + 
(\mathsf{U}+\mathsf{C}) / \delta \eps .
$

\end{enumerate}
\end{proposition}

The generic bound of $O(1/ \delta \eps)$ in {\bf Search Algorithm~3} on the hitting time
is not always optimal, which in some cases, for example 
in the 2-dimensional grid, can be significantly smaller. 

\section{Quantum analogue of a classical Markov chain}
\label{sec-quantum-mc}
We define a quantum analogue of an arbitrary irreducible Markov
chain~$P$ as it is given by Magniez et al.~\cite{MNRS}. This definition is
based on and slightly extends the concept of quantum Markov chain
due to Szegedy~\cite{Szegedy04}. The latter was inspired by an earlier
notion of quantum walk due to Ambainis~\cite{Ambainis04}. We also
point out that a similar process on regular graphs was studied by
Watrous~\cite{Watrous01}.

The quantum walk may be thought of as a walk on the {\em edges\/} of the
original Markov chain, rather than on its vertices. Thus, its state space
is a vector subspace of~$\aitch = \complex^{X \times X}
\isomorphic \complex^X \tensor \complex^X$.
For a state~$\ket{\psi} \in \aitch$,
let~$\Pi_\psi = \ket{\psi}\bra{\psi}$ denote the orthogonal projector
onto~$\Span(\ket{\psi})$, and let $\reflex(\psi)= 2 \Pi_\psi - \id $
denote the reflection through the line generated by $\ket{\psi}$,
where~$\id$ is the identity operator on~${\aitch}$.
If $\kay$ is a subspace of $\aitch$ spanned by a set of mutually orthogonal
states $\{\ket{\psi_i} : i \in I\}$, then let
$\Pi_{\kay} = \sum_{i \in I} \Pi_{\psi_i} $ be the orthogonal projector
onto $\kay$, and let 
$\reflex(\kay) = 2 \Pi_{\kay} - \id $
be the reflection through $\kay$.
Let $\ay=\Span( \ket{x}\ket{p_x} : x\in X)$ and 
$\bee=\Span( \ket{p^*_y}\ket{y} : y\in X)$ be vector subspaces
of~$\aitch$, 
where
\centermath{$
\ket{p_x} =  \sum_{y\in X} \sqrt{p_{xy}} \, \ket{y} \quad\textrm{and} \quad
\ket{p^*_y}  = \sum_{x\in X} \sqrt{p^*_{yx}} \, \ket{x}.
$}

\begin{definition}[Quantum walk]
The unitary operation $W(P)  =  \reflex(\bee) \cdot \reflex(\ay)$ defined on $\aitch$ by 
is called the {\em quantum walk\/} based on the classical chain~$P$.
\end{definition}

Let us give some motivations for this definition. Classical random walks do not quantize
in the space of the vertices. The standard way advocated by several 
papers (see the survey~\cite{Kem2}) is to extend
the vertex space $X$ by a coin space $C$, and define the state space of the walk as $X \times C$.
Then a step of the walk is defined as the product of two unitary operations. 
The first one is the {\em flip} operation $F$  controlled by the vertex state,
which means that  for every $x \in X$, it performs a unitary coin flip $F^x$ on the
states $\{ \ket{x,c} : c \in C \}$. 
For $d$-regular undirected graphs, $C$ can be taken as the set $\{ 1, \ldots , d \}$,
and in that case the coin flip $F^x$  is independent from $x$.
The second one is the {\em shift} operation $S$ which is controlled by the coin state, and takes a
vertex to one if its neighboring vertices. 
For $d$-regular graphs the simplest way to define it is via a labeling of the directed edges
by the numbers between 1 and $d$ such that for every $1 \leq i \leq d$, the directed edges labeled
by $i$ form a permutation. Then,
if the coin state is $i$, the new vertex is 
the $i^{\rm th}$ neighbor according  to the labeling.
For general walks it is practical to take the coin space also to be $X$,
then the state space of the walk corresponds naturally to the directed edges of the graph.
In this case there is a symmetry between the two spaces, and
the shift operation simply exchanges the vertices, that is 
$S \ket{x,y} = \ket{y,x}$, for every $x,y \in X$. 

Let us pause here for a second and consider how a classical walk
defined by some Markov chain $P$ can be thought of
as a walk on the directed edges of the graph (instead of the vertices). Let's think about an
edge $(x,u)$ as the state of the walk $P$ being at $x$, where the previous state was $u$.
According to  this interpretation, 
in one step the walk on edges should move 
from state $(x,u)$ to state $(y,x)$ with probability $p_{xy}$. This move can be
accomplished by the stochastic flip operation $F$ 
controlled by the left end-point of the edge, where  
$F^x_{uy} = p_{xy}$ for all $x,u,y \in X$, followed by the shift $S$
defined previously.
If we define the flip $F'$ as $F$ but controlled by the right end-point of the edge, then
it is not hard to see that $SFSF = F'F$.
Therefore one can get rid of
the shift operations, and two steps of the walk can be accomplished by
two successive flips where the control and the
target registers alternate.
 
Coming back to the quantization of classical walks, we thus want to find
unitary coin flips which mirror the walk $P$, and which alternately mix the 
right end-point of
the edges over the neighbors of the left end-point, and then the
left end-point of the edges over the neighbors of the new right
end-point. 
The reflections  $ \reflex(\ay) $ and $ \reflex(\bee)$ are natural choices for that.
They are also generalizations of the Grover diffusion operator~\cite{Grover96}.
Indeed, when the
transition to each neighbor is equally likely, they
correspond exactly to Grover diffusion. In Szegedy's 
original definition 
the alternating reflections were 
$ \reflex(\ay) $ and $ \reflex(\beee)$ with
$\beee=\Span( \ket{p_y}\ket{y} : y\in X)$, mirroring faithfully
the classical edge based walk. The reason why 
the MNRS quantization chooses every second step
a reflection based on the reversed walk $P^*$ is explained now.

The eigen-spectrum of the transition matrix~$P$ plays an important
role in the analysis of a classical Markov chain. Similarly, the
behaviour of the quantum process~$W(P)$ may be inferred from its
spectral decomposition.
The reflections through subspaces~$\ay$ and $\bee$
are (real) orthogonal transformations, and so is their product~$W(P).$ 
An orthogonal matrix may be
decomposed into a direct sum of the identity, reflection through the
origin,
and two-dimensional rotations over orthogonal
vector subspaces~\cite[Section 81]{Halmos74}. These 
subspaces and the corresponding eigenvalues are revealed by the
singular value decomposition of the product~$\Pi_\ay \Pi_\bee$ of the
orthogonal projection operators onto the 
subspaces~$\ay$ and $\bee$. Equivalently, as done by Szegedy, one can
consider the singular values of the {\em
discriminant} matrix~$D(P)=(\sqrt{p_{xy}p^*_{yx}})$. Since 
$\sqrt{p_{xy}p^*_{yx}} = \sqrt{\pi_x} p_{xy}/ \sqrt{\pi_y}$,
we have 
\centermath{
$D(P) \quad = 
\quad \diag(\pi)^{1/2} \cdot P \cdot \diag(\pi)^{-1/2},$
}
where~$\diag(\pi)$ is the invertible diagonal matrix with the
coordinates of the distribution~$\pi$ in its diagonal.  
Therefore $D(P)$ and $P$ are similar, and their spectra are the same.
When $P$ is reversible then $D(P)$ is symmetric, and its singular values
are equal to the absolute values of its eigenvalues.
Thus, in that case we only have to study the spectrum of $P$.

Since the singular values of~$D(P)$
all lie in the range~$[0,1]$,
they can be expressed 
as~$\cos\theta$, for some angles~$\theta \in
[0,\tfrac{\pi}{2}]$.
The following theorem of Szegedy 
relates the
singular value decomposition of~$D(P)$ to the spectral decomposition
of~$W(P)$.
\begin{theorem}[Szegedy~\cite{Szegedy04}]
\label{thm-spectrum}
Let $P$ be an irreducible Markov chain, and
let $\cos\theta_1,\ldots,\cos\theta_l$ be an enumeration of those
singular values (possibly repeated) of $D(P)$ that lie in the open
interval~$(0,1)$.  Then
the exact description of the spectrum of $W(P)$ on $\ay+\bee$ is:
\begin{enumerate}\reduceitemm
\item On $\ay + \bee$ those eigenvalues of $W(P)$
that have non-zero imaginary part are exactly $e^{\pm 2 \complexi
\theta_1},\ldots,e^{\pm 2 \complexi \theta_l}$, with the same
multiplicity.

\item On $\ay\cap\bee$ the operator $W(P)$ acts as the
identity~$\id$. $\ay\cap\bee$ is spanned by the left
(and right) singular vectors of $D(P)$ with singular value~$1$.
\comment{In other words, the dimension of $\ay\cap\bee$
is the multiplicity of the singular value $1$ in $D(P)$.}

\item On $\ay \cap \bee^\perp$ and
$\ay^\perp\cap\bee$ the operator $W(P)$ acts as $-\id$.
$\ay\cap\bee^\perp$ (respectively,
$\ay^\perp\cap\bee$) is spanned by the left
(respectively, right) singular vectors of $D(P)$ with singular
value~$0$.

\end{enumerate}\reduceitemm
\end{theorem}

Let us now suppose in addition that $P$ is ergodic and reversible.
As we just said,
reversibility implies that the singular values of $D(P)$ 
are equal to the 
absolute values of the eigenvalues of $P$. From the ergodicity it also follows that
$D(P)$ has a unique
singular vector with singular value 1. We have  therefore the following corollary.

\begin{corollary}
\label{cor-spectrum}
Let $P$ be an ergodic and reversible Markov chain.
Then, on $\ay+\bee$ the spectrum of $W(P)$ can be characterized as:
$$\ket{\pi}   =  \sum_{x \in X} \sqrt{\pi_x} \, \ket{x}\ket{p_x} =
\sum_{y \in X} \sqrt{\pi_y} \, \ket{p_y^*}\ket{y}
$$
is the unique 1-eigenvector, 
$e^{\pm 2 \complexi \theta}$ are eigenvalues
for every singular value
$\cos \theta \in (0,1)$ of $D(P)$, 
and all the remaining eigenvalues are -1.

\end{corollary}

The {\em phase gap} $\Delta(P) = \Delta$ of $W(P)$ is 
defined as $2\theta$, where $\theta$ is the smallest
angle in $ (0,\tfrac{\pi}{2}]$ such that $\cos \theta$ is a singular value of $D(P)$.
This definition is motivated by the previous theorem and corollary:
the angular distance of $1$ from any other eigenvalue 
of $W(P)$ on $\ay+\bee$ is at least $\Delta$.
When $P$ is ergodic and reversible, there is a 
quadratic relationship  between the phase gap~$\Delta$ of the
quantum walk~$W(P)$ and the eigenvalue gap~$\delta$ of the classical Markov chain~$P$,
more precisely
$\Delta \geq 2 \sqrt{\delta}$.
Indeed, let $\delta \in (0,\tfrac{\pi}{2}]$ such that 
$\delta = 1- \cos \theta$ and $\Delta = 2 \theta$.
The following (in)equalities can easily be checked:
$ \Delta
     \geq  \abs{1 - \e^{2 \complexi \theta }}
 = 2 \sqrt{1- \cos^2 \theta} 
    \geq  2\sqrt{\delta}.
$
The origin of the quadratic speed-up due to quantum walks may be
traced to this phenomenon.

\section{Quantum search algorithms}
\label{sec-q-analogue}
As in the classical case, the quantum search algorithms look for a marked element 
in a finite set $X$. We suppose that the elements of $X$ are coded by binary strings and that
$\bar{0}$, the everywhere 0 string is in $X$.
A data structure attached to both vertex registers is maintained during the algorithm.
Again, three types of cost will be distinguished, generalizing those of the classical search.
In all quantum
search algorithms the overall complexity 
is of the order of these specific costs, which justifies their choices.
The operations not
involving manipulations of the data will be charged 
at unit cost.
For the sake of simplicity,  we do not  formally include the data
into the description of the unitary operations defining the costs.
The initial state of the algorithm is explicitly related to the stationary
distribution $\pi$ of $P$.

\begin{itemize}\reduceitem

\item[] \textbf{(Quantum) Setup cost $\mathsf{S}$:} The cost
for constructing the state $\sum_{x \in X}
\sqrt{\pi_x}\ket{x}\ancilla$ with data.

\item[] \textbf{(Quantum) Update cost $\mathsf{U}$:} The cost to
realize any of the unitary transformations and inverses with data
\smallskip\\
\centerline{$\ket{x} \ancilla
    \quad \mapsto \quad \ket{x}
                        \sum_{y \in X}\sqrt{p_{xy}}\ket{y},$}
\centerline{$\ancilla \ket{y}
    \quad \mapsto \quad 
                        \sum_{x \in X}\sqrt{p^*_{yx}}\ket{x}\ket{y}.$}
\smallskip
\item[] \textbf{(Quantum) Checking cost $\mathsf{C}$:} 
The cost to realize the unitary transformation with data, that
maps $\ket{x}\ket{y}$ to $-\ket{x}\ket{y}$ if $x\in M$,
and leaves it unchanged otherwise.
\end{itemize}
In the checking cost we could have included the  cost of the unitary transformation
which realizes a phase flip also when $y \in M$, our choice was made just for simplicity.
Observe that the quantum walk $W(P)$ with data can be implemented
at cost $4\mathsf{U}+2$.
Indeed,
the reflection~$\reflex(\ay)$ is implemented 
by mapping states~$\ket{x}\ket{p_x}$ to~$\ket{x}\ancilla$, 
applying 
$\reflex ( \complex^X\otimes\ancilla)$,
and undoing the first transformation.
In our accounting we charge unit cost for the second step since it does not depend
on the database. Therefore the implementation of $\reflex(\ay)$ is of cost
$2\mathsf{U}+1$.
The reflection $\reflex(\bee)$ may be implemented similarly.

Let us now describe how the respective algorithms of Grover, Ambainis and Szegedy are
related to the classical search algorithms of Section~\ref{classical}.
We suppose that $\eps$, a lower bound on the proportion of marked 
elements is known in advance,
though the results remain true even if it is not the case.
Grover search (which we discuss soon in detail) is the quantum analogue of 
\textbf{Search Algorithm~1}.
\begin{theorem}[Grover~\cite{Grover96}]
\label{thm-grover}
There exists a quantum algorithm which with high probability finds a marked element,
if there is any, at cost of order $\frac{\mathsf{S} + \mathsf{C}}{\sqrt{\eps}}$.
\end{theorem}

In the original application of Grover's result to unordered search there is no data structure
involved, therefore $\mathsf{S} + \mathsf{C} = O(1)$, and the cost is of order $\frac{1}{\sqrt{\eps}}$.

The algorithm of Ambainis is the quantum analogue of \textbf{Search Algorithm~2}
in the special case of the walk on the Johnson graph and for some specific marked sets.
Let us recall that
for $0 < r \leq n/2$, the vertices of the Johnson graph $J(n,r)$ are 
the subsets of $[n]$ of size $r$, and there is an edge between
two vertices if the  size of  their symmetric difference is 2. In other words, 
two vertices are adjacent if by deleting an element from the first one and adding a new element
to it we get the second vertex.
The eigenvalue gap $\delta$ of the symmetric walk on $J(n,r)$ is $n/r(n-r) = \Theta(1/r)$. If the set
of marked vertices in $J(n,r)$ is either empty, or it consists of 
vertices that contain a fixed subset of 
constant size~$k \leq r$ then $\eps = \Omega(\frac{r^k}{n^k})$.

\begin{theorem}[Ambainis~\cite{Ambainis04}]
\label{thm-ambainis}
Let~$P$ be the random walk on the Johnson graph 
$J(n,r)$ where~$r = \order(n)$, and
let~$M$ be either empty,
or the class of vertices that contain a fixed subset of 
constant size~$k \leq r$. 
Then there is a quantum algorithm that finds,
with high probability, 
the~$k$-subset if $M$ is not empty at cost of
order
$
\mathsf{S} + \frac{1}{\sqrt{\eps}} ( \frac{1}{\sqrt{\delta}}
\mathsf{U}+\mathsf{C}).
$
\end{theorem}

Szegedy's algorithm is the quantum analogue of \textbf{Search Algorithm~3} for the class of
ergodic and symmetric Markov chains. His algorithm is therefore more general than 
the one of Ambainis
with respect to the class of Markov chains and marked sets it can deal with. 
Nonetheless, the approach of Ambainis has its own advantages: it is of smaller cost when
$\mathsf{C}$ is substantially greater than $\mathsf{U}$, and it also finds a marked element.

\begin{theorem}[Szegedy~\cite{Szegedy04}
] 
\label{thm-szegedy}
Let $P$ be an ergodic and symmetric Markov chain.
\suppress{
Assume the hypotheses of
Proposition~\ref{thm-classical-search}.  
}
There exists a quantum
algorithm that determines, with high probability, if $M$ is non-empty at cost of order
$
\mathsf{S} + \frac{1}{\sqrt{\delta\eps}} (\mathsf{U}+\mathsf{C}).
$
\end{theorem}

The MNRS algorithm is a quantum analogue of {\bf Search Algorithm 2} for
ergodic and reversible Markov chains. 
It generalizes the algorithms of Ambainis an Szegedy, and it
combines their benefits in terms of being
able to find marked elements, incurring the smaller cost of the two,
and being applicable to a larger class of Markov chain.

\begin{theorem}[Magniez et al.~\cite{MNRS}]
\label{thm-with-log-factor}
Let $P$ be an ergodic and reversible Markov
chain, and let $\eps > 0$ be a lower bound on the probability that
an element chosen from the stationary distribution of $P$ is marked
whenever~$M$ is non-empty.  Then, there exists  a quantum algorithm which 
finds, with high probability, 
an element of $M$ if there is any at cost of
order
$
\mathsf{S} + \frac{1}{\sqrt{\eps}} 
( \frac{1}{\sqrt{\delta}}
\mathsf{U}+\mathsf{C} 
).$
\end{theorem}

There is an additional feature of Szegedy's algorithm
which doesn't fit into the MNRS algorithmic paradigm.
In fact, the quantity
$\frac{1}{\sqrt{\delta\eps}} $ in Theorem~\ref{thm-szegedy}
can be replaced by the square root of the classical hitting time~\cite{Szegedy04}.
The search algorithm for the 2-dimensional grid 
obtained this way, and the one given in~\cite{AmbainisKR05}
have smaller complexity than what follows from Theorem~\ref{thm-with-log-factor}.

\section{The MNRS search algorithm}
\label{sec-walk-search}

We give a high level description of the MNRS search algorithm.
Assume that 
$M \not= \emptyset$. 
Let $\emm=\complex^{M \times
X}$ denote the marked subspace, that is the subspace with marked items in the first register.
The purpose of the algorithm is to approximately transform the initial state~$\ket{\pi}$ to the target
state~$\ket{\mu}$, which is the normalized projection of~$\ket{\pi}$
onto 
$\emm$:  

$$\ket{\mu} \quad = \quad \frac{\Pi_\emm\ket{\pi}}{\norm{\Pi_\emm\ket{\pi}}}
          \quad = \quad \frac{1}{\sqrt{\eps}} \sum_{x \in M} \sqrt{\pi_x} \,
                 \ket{x}\ket{p_x},
$$
where~$\eps=\norm{\Pi_\emm\ket{\pi}}^2 = \sum_{x \in M} \pi_x$
is the probability of the set~$M$ of marked states under the stationary
distribution~$\pi$. Let us recall that
Grover search~\cite{Grover96} solves this problem via the
iterated use of the rotation
$
\reflex(\pi) \cdot \reflex(\mu^\perp)
$
in the two-dimensional real subspace~$\ess =
\Span(\ket{\pi},  \ket{\mu})$,
where $\ket{\mu^\perp}$ is the state 
in~$\ess$ orthogonal to~$\ket{\mu}$ making some acute angle $\varphi$
with~$\ket{\pi}$.
The angle~$\varphi$
is given by~$\sin \varphi =
\braket{\mu}{\pi} = \sqrt{\eps}$. 
Then $
\reflex(\pi) \cdot \reflex(\mu^\perp)
$ 
is a rotation by $2\varphi$ within the
space~$\ess$, and therefore $\Order(1/\varphi) =
\Order(1/\sqrt{\eps})$ iterations  of this rotation, starting with~$\ket{\pi}$,
approximates well~$\ket{\mu}$. The MNRS search basically follows this idea.

Restricted to the subspace~$\ess$, the operator $\reflex(\mu^\perp)$ is identical to
$- \reflex(\emm)$. Therefore, if the state of the algorithm
remains close to the subspace~$\ess$ throughout, it can be
implemented at the price of the checking cost.
The reflection $\reflex(\pi)$ is
computationally harder to perform. 
The idea is to apply the phase estimation algorithms of 
Kitaev~\cite{Kitaev96} and Cleve at al.~\cite{CleveEMM98} to
$W(P)$. Corollary~\ref{cor-spectrum}
implies that ~$\ket{\pi}$ is the only 1-eigenvector of $W(P)$, and all the
other eigenvectors have phase at least $\Delta$. Phase estimation approximately resolves
any state $\ket{\psi}$ in $ \ay + \bee$ along the eigenvectors of $W(P)$, and thus
distinguishes $\ket{\pi}$ from all the others. Therefore it is possible to flip the phase
of all states with a non-zero estimate of the phase, that is simulate
the effect of the operator $\reflex(\pi)$
in $ \ay + \bee$. The following result of~\cite{MNRS} resumes this discussion:

\begin{theorem}
\label{thm-diffusion1}
There exists a uniform family of quantum circuits $R(P)$ 
that uses $O( k \log ( \Delta^{-1}) ) $ additional qubits
and satisfies the following properties:
\begin{enumerate}\reduceitemm

\item
It makes $O(k \Delta^{-1})$ calls to the 
controlled quantum walk $\mathrm{c-}W(P)$ and its
inverse.
\item 
$R(P)\ket{\pi} = \ket{\pi}$.

\item If~$\ket{\psi} \in \ay + \bee$ is
orthogonal to~$\ket{\pi}$,
then~$\norm{ (R(P) + \id)  \ket{\psi} } \leq 2^{-k}$.

\end{enumerate}\reduceitemm
\end{theorem}
The essence of the 
MNRS search algorithm 
is the following simple procedure  that
satisfies the conditions of 
Theorem~\ref{thm-with-log-factor}, but with a slightly higher 
complexity, of the order of 
$
\mathsf{S} + \frac{1}{\sqrt{\eps}} 
( \frac{1}{\sqrt{\delta}}
\log\frac{1}{\sqrt{\eps}} 
\mathsf{U}+\mathsf{C} 
).$ Again, we suppose that $\eps$ is known.
\encadre{ $\textbf{Quantum Search}(P)$
\begin{enumerate}\reduceitem
\item Start from the initial state $\ket{\pi}$.
\item Repeat $\Order(1/\sqrt{\eps})$--times:
\begin{enumerate}\reduceitemm
\item For any basis vector $\ket{x}\ket{y}$,
flip the phase  if $x\in M$.
\item Apply circuit $R(P)$ of Theorem~\ref{thm-diffusion1} with $k
= \Order( \log(1/\sqrt{\eps}))$.
\end{enumerate}
\item Observe the first register and 
output if it is in $ M$.
\end{enumerate}
} 
To see the correctness,
let $\ket{\phi_i}$ be the result of
$i$
iterations of 
$\reflex(\pi)\cdot\reflex(\mu^\perp)$
applied to $ \ket{\pi}$, 
and let
$\ket{\psi_i}$ be the result of
$i$  iterations of step (2) in 
\textbf{Quantum Search$(P)$} applied to $ \ket{\pi}$. 
It is not hard to show by induction 
on $i$, using a hybrid argument as in~\cite{BennettBBV97,Vazirani98b},
that
$\norm{\ket{\psi_i} - \ket{\phi_i}} \leq O(i 2^{-k})$.
This implies that
$\norm{\ket{\psi_k} - \ket{\phi_k}}$ is an arbitrarily small constant when
$k$ is chosen to be 
$ \Order( \log(1/\sqrt{\eps}))$
and therefore the success probability
is arbitrarily close to 1.

The cost of the procedure is simple to analyze. Initialization costs 
$\mathsf{S} + \mathsf{U}$,
and in each iteration the single phase flip costs $\mathsf{C}$. In the circuit
$R(P)$, the controlled quantum walk and its inverse can be implemented,
similarly to $W(P)$, at cost $4\mathsf{U}+2$, simply by
controlling $\reflex ( \complex^X\otimes\ancilla)$ and $\reflex ( \ancilla\otimes\complex^Y)$.
The number of steps of the controlled
quantum walk and its inverse is 
$\Order(  (1/{\Delta}) \log ( 1/ \sqrt{\eps}) )$. 
Since $\Delta \geq  2 \sqrt{\delta}$,
this finishes the cost analysis.
Observe that the $ \log ( 1/ \sqrt{\eps})$-factor in the update cost was necessary
for reducing the error of the approximate reflection operator.
In~\cite{MNRS} it is described how it can be eliminated
by adapting the recursive amplitude amplification algorithm of
H{\o}yer et al.~\cite{HoyerMW03}

\section{Applications}
We give here a few examples where the quantum search algorithms, in particular the MNRS
algorithm can be applied successfully.
All examples will be described in the query model of computation. Here the input is given
by an oracle, a query can be performed at unit cost, and all other computational steps are free.
A formal description of the model can be found for example in~\cite{MagniezN05}
or~\cite{MagniezSS05}. 
In fact, in almost all cases, the circuit complexity of the algorithms given will be of the order of
the query complexity, with additional logarithmic factors.
\subsection{Grover search}
As a first (and trivial) application, we observe that Grover's algorithm~\cite{Grover96}
for the unordered search problem is a special case of Theorem~\ref{thm-with-log-factor}.

\begin{quote}
\textsc{Unordered Search}\\
\textit{Oracle Input:} A boolean function $f$ defined on $[n]$.\\
\textit{Output:} An element $i \in [n]$ such that  
$f(i) = 1$.
\end{quote}

\begin{theorem}
\label{unordered}
\textsc{Unordered Search} can be solved with high probability
in quantum query complexity ${O}((n/k)^{1/2})$, where $|\{i \in [n]~:~f(i)=1\}| = k$.
\end{theorem}
\begin{proof}
Consider the symmetric random walk in the complete graph
on $n$ vertices, where an element $v$ is marked if $f(v) = 1$. The eigenvalue gap of
the walk is $1 - \frac{1}{n-1}$, and the probability $\eps$ 
that an element is marked is $k/n$.
There is no data structure involved in the algorithm, the setup, update and  checking costs are 1.
\end{proof}

\subsection{Johnson graph based algorithms}
All these examples are based on the symmetric walk in the Johnson graph $J(n,r)$,
with eigenvalue gap $\Theta (1/r)$.
\subsubsection{Element Distinctness}
This is the original problem for which Ambainis
introduced the quantum walk based search method~\cite{Ambainis04}. 

\begin{quote}
\textsc{Element Distinctness}\\
\textit{Oracle Input:} A function $f$ defined on $[n]$. \\
\textit{Output:} A pair of distinct elements $i,j \in [n]$ such that  
$f(i) = f(j)$ if there is any,
otherwise reject.
\end{quote}

\begin{theorem}
\label{distinct}
\textsc{Element Distinctness} can be solved with high probability
in quantum query complexity ${O}(n^{2/3})$.
\end{theorem}

\begin{proof}
A vertex $ R \subseteq [n]$ of size $r$ is marked if there exist $i \neq j \in R$ such $f(i) = f(j)$.
The probability $\eps$ that an element is marked, if there is any, is
in $\Omega ((r/n)^2)$. 
For every $R$,
the data 
is defined as $\{(v,f(v)):v\in R\}$. Then the setup cost 
is in $O(r)$, the
update cost  
is $O(1)$, and the checking cost is 0.
Therefore the overall cost is
$O( r + n/r^{1/2})$ which is $O(n^{2/3})$ when $r = n^{2/3}$.
This upper bound is tight, the $\Omega (n^{2/3})$ lower bound is due to 
Aaronson and Shi~\cite{AaronsonS04}.
\end{proof}

\subsubsection{Matrix Product Verification}
This problem was studied by Buhrman and 
Spalek~\cite{BuhrmanS06}, and the algorithm in the query model
is almost identical to the previous one.

\begin{quote}
\textsc{Matrix Product Verification}\\
\textit{Oracle Input:} Three $n \times n$ matrices $A,B$ and $C$. \\
\textit{Output:} Decide if $A B = C$ and 
in the negative case find indices $i,j$ such that $(AB)_{ij} \neq C_{ij}$.
\end{quote}

\begin{theorem}
\label{matrix}
\textsc{Matrix Product Verification} can be solved with high probability
in quantum query complexity ${O}(n^{5/3})$.
\end{theorem}

\begin{proof}
For an $n \times n$ matrix $M$ and a subset of indices $R \subseteq [n]$,
let $M|_R$ denote the $|R| \times n$ submatrix of $M$ corresponding to the rows 
restricted to $R$. The submatrices $M|^R$ and $M|_R^R$ are defined similarly, when 
the restriction concerns the  columns and both the rows and columns.
A vertex $R \subseteq [n]$ is marked if 
there exist $i,j \in R$ such that $AB_{ij} \neq C_{ij}$. 
The probability of being marked, if there is such an element,
is in $\Omega ((r/n)^2)$. 
For every $R$, 
the data  
is defined as
the set of entries in $A|_R, B|^R$ and $C|_R^R$.  Then the setup cost
is in $O(rn)$, the
update cost is $O(n)$, and the checking cost is 0. Therefore 
the overall cost is $O(n^{5/3})$ when $r = n^{2/3}.$
The best known lower bound is $\Omega (n^{3/2})$, and in ~\cite{BuhrmanS06}
an ${O}(n^{5/3})$ upper bound was also proven for the time complexity with a somewhat more
complicated argument.

\end{proof}

\subsubsection{Restricted Range Associativity }
The problem here is to decide if a binary operation $\circ$ is associative. The only algorithm known 
for the general case is Grover search, but D\"orn and Thierauf~\cite{DT} 
have proved that when the range of the operation is restricted, Theorem~\ref{thm-with-log-factor} 
(or Theorem~\ref{thm-ambainis})
give a non-trivial bound. A triple $(a,b,c)$  is called {\em non associative}
if $(a \circ b) \circ c \neq a \circ (b \circ c)$.
\begin{quote}
\textsc{Restricted Range Associativity}\\
\textit{Oracle Input:} A binary operation $\circ : [n] \times [n] \rightarrow [k]$
where $k \in O(1)$. \\
\textit{Output:} A non associative triple $(a,b,c)$ if there is any,
otherwise reject.
\end{quote}

\begin{theorem}
\label{associativity}
\textsc{Restricted Range Associativity} can be solved with high probability
in quantum query complexity ${O}(n^{5/4})$.
\end{theorem}

\begin{proof}
We say that $R \subseteq [n]$ is marked
if there exist $a,b \in R$ and $c \in [n]$ such that $(a,b,c)$ is non associative.
Therefore 
$\eps$ is
in $\Omega ((r/n)^2)$, if there is a marked element.
For every $R \subseteq [n]$, the data structure is defined as
$ \{(a,b,a \circ b) : a,b \in R \cup [k] \}$. Then the setup cost is
$O((r+k)^2 = O(r^2)$ and the update cost is $O(r+k) = O(r)$.
Observe that if $b \in R$ and $c \in [n]$ are fixed, then computing
$(a \circ b) \circ c$ and $a \circ (b \circ c)$ for all $a \in R$ requires at most $k+1$ queries
with the help of the data structure. Thus using Grover search to find $b$ and $c$, the checking cost is
$O(k \sqrt{rn}) = O(\sqrt{rn})$. The overall complexity is then
${O}(r^2+\tfrac{n}{r}(\sqrt{r} r + \sqrt{rn}))$
which is ${O}(n^{5/4})$ when $r= \sqrt{n}$. 
The best lower bound known both in the restricted range and the general case is $\Omega(n)$.

\end{proof}

\subsubsection{Triangle}
In an undirected graph $G$,
a complete subgraph on three vertices is called a {\em triangle}. The algorithm of
Magniez et al.~\cite{MagniezSS05}  for finding a triangle uses
the algorithm for \textsc{Element Distinctness} in the checking procedure.

\begin{quote}
\textsc{Triangle}\\
\textit{Oracle Input:} The adjacency matrix $f$ of a graph $G$
on vertex set $[n]$.\\
\textit{Output:} A triangle if there is any,
otherwise reject.
\end{quote}

\noindent
\begin{theorem}
\label{triangle}
\textsc{Triangle} can be solved with high probability
in quantum query complexity ${O}(n^{13/10})$.
\end{theorem}

\begin{proof}
We show how to find the edge of a triangle, if there is any, in 
query complexity ${O}(n^{13/10})$.
This implies the theorem 
since given such an edge, Grover search finds the third vertex
of the triangle with $O(n^{1/2})$ additional queries.

An element $R \subseteq [n]$ is marked if it contains a triangle edge.
The probability $\eps$ that an element is marked is
in $\Omega ((r/n)^2)$, if there is a triangle.
For every $R$,  
the data structure is the adjacency matrix of the subgraph induced by $R$.
defined as $\{(v,f(v)):v\in R\}$. 
Then the setup cost 
is $O(r^2)$, and the
update cost  
is $O(r)$. 
The interesting part of the algorithm is the checking procedure, 
which is a quantum walk based search itself, and the claim is
that it can be done at cost  
$O(\sqrt{n}\times r^{2/3})$.

To see this, let $R$ be a set of $r$ vertices such that 
the graph $G$ restricted
to $R$ is explicitly known,
and for which we would like to decide if it is marked. Observe that $R$ is marked exactly when 
there exists a vertex $v \in [n]$ such that $v$ and an edge in $R$ form a triangle.
Therefore for every vertex $v$, one can define a 
secondary search problem on $R$ via the  boolean oracle $f_v$,
where for every $u \in R$, by definition $f_v(u) = 1$ if $\{u,v\}$ is an edge. The output of the 
problem is by definition positive if there is an edge $\{u,u'\}$ such that $f_v(u) = f_v(u') = 1$.
To solve the problem we consider the Johnson graph $J(r, r^{2/3})$,
and look for a subset which contains
such an edge. In that search problem both the probability 
of being marked and the eigenvalue gap 
of the underlying Markov chain are in $\Omega ( r^{-2/3})$.
The data associated with a subset of $R$ is just the values of $f_v$ at its 
elements. Then the setup cost is $r^{2/3}$, the update cost is $O(1)$, and the checking cost is 0.
Therefore by Theorem~\ref{thm-with-log-factor} 
the cost of solving a secondary search problem is in $O(r^{2/3})$.
Finally the checking procedure of the original search problem consists 
of a Grover search for a vertex $v$
such that the secondary search problem defined by $f_v$ has a positive outcome.
Putting things together, the problem can be solved in quantum query complexity
${O}(r^2+\tfrac{n}{r}(\sqrt{r}\times r + \sqrt{n}\times r^{2/3}))$
which is ${O}(n^{13/10})$ when $r=n^{3/5}$.
The best known lower bound for the problem is $\Omega(n)$.
\end{proof}

\subsection{Group commutativity}
The problem here is to decide if a group multiplication is commutative in the
(sub)group generated by some set of group elements. It was defined and studied
in the probabilistic case by Pak~\cite{Pak}, the quantum algorithm
is due to Magniez and Nayak~\cite{MagniezN05}.

\begin{quote}
\textsc{Group Commutativity}\\
\textit{Oracle Input:} The multiplication operation $\circ $ for a finite group 
whose base set contains $[n]$.\\
\textit{Output:} A non commutative couple $(i,j) \in [n] \times [n]$ if $G$,
the (sub)group generated by $[n]$, is non-commutative,
otherwise reject.
\end{quote}

\begin{theorem}
\label{commutativity}
\textsc{Group Commutativity} can be solved with high probability
in quantum query complexity ${O}(n^{2/3} \log n)$.
\end{theorem}

\begin{proof}
For $0 < r < n$ let $S(n,r)$ be the set of all $r$-tuples of distinct elements from $[n]$.
For $u = (u_1, \ldots , u_r)$ in $S(n,r)$, we set $\bar{u} = u_1 \circ \cdots \circ u_r$.
We define a random walk over $S(n,r)$. Let $u = (u_1, \ldots , u_r)$ be the current vertex.
Then with probability 1/2 stay at $u$, and with probability 1/2 
pick uniformly random $i \in [r]$ and  $j \in [n]$.
If $j = u_m$ for some $m$ then exchange $u_i$ and $u_m$, otherwise set $u_i = j$.
The random walk $P$ at the basis of the quantum algorithm is over $S(n,r) \times S(n,r)$,
and it  consists of two independent simultaneous copies of the above walk. The
stationary distribution of $P$ is the uniform distribution, and it is proven 
in~\cite{MagniezN05} that its eigenvalue gap is $\Omega (1/(r \log r))$. 

A vertex $(u,v)$ is marked if $\bar{u} \circ \bar{v} \neq \bar{v} \circ \bar{u}$. It is proven again
in~\cite{MagniezN05} that when $G$ is non-commutative and $r \in o(n)$, 
then the probability $\eps$ that an element is marked is
$\Theta (r^2/n^2)$. For $u \in S(n,r)$ let $T_u$ be the balanced binary tree with $r$ leaves that
are labeled from left to right by $u_1, \ldots, u_r$, and where each internal node is
labeled by the product of the labels of its two sons. For every vertex $(u,v)$ the data consists
of $(T_u,T_v)$. Then the setup cost is $r$, and the update cost is $O(\log r)$ for recomputing
the leaf--root paths.
The checking cost is simply
$2$ for querying $\bar{u} \circ \bar{v}$ and $\bar{v} \circ \bar{u}$.
Therefore the query complexity to find a marked element is
${O}(r+\frac{n}{r}(\sqrt{r \log r} \log r + 1))$
which is ${O}(n^{2/3} \log n)$ when $r= n^{2/3} \log n$. Once a marked element is found,
Grover search yields a non-commutative couple at cost $O(r)$.
In~\cite{MagniezN05} an $\Omega (n^{2/3})$ lower bound is also proven. 
And, it turns out that a Johnson graph based walk can be applied to this problem 
too~\cite{MagniezN08}, yielding an algorithm of complexity ${O}((n \log n)^{2/3} )$.

\end{proof}

\section*{Acknowledgment}
I would like to thank Fr\'ed\'eric Magniez, Ashwin Nayak and J\'er\'emie Roland, my coauthors in \cite{MNRS}, for letting me to include here
several  ideas which were developed during that work, and for numerous helpful  suggestions.

\end{document}